\documentclass[letterpaper, 10 pt, conference]{ieeeconf}  %

\IEEEoverridecommandlockouts                              %

\usepackage{graphics} %
\usepackage[pdftex]{graphicx}
\usepackage{epsfig} %
\usepackage{mathptmx} %
\usepackage{amsmath} %
\usepackage{amssymb}  %
\usepackage{newtxtext,newtxmath}
\usepackage{mathtools}
\usepackage{algorithmic}
\usepackage{algorithm}
\usepackage{siunitx}
\usepackage{hyperref}
\graphicspath{{figures/}}

\usepackage{xcolor}
\usepackage{environ}

\newcommand{\past}{{\mathrm{p}}}
\newcommand{\future}{{\mathrm{f}}}
\newcommand{\hp}{\mathrm{h}_{\mathrm{p}}}
\newcommand{\hf}{\mathrm{h}_{\mathrm{f}}}
\newcommand{\utp}{u_t^{\past}}
\newcommand{\utf}{u_t^{\future}}
\newcommand{\ytp}{y_t^{\past}}

\newcommand{\wtf}{w_t^{\future}}

\newcommand{\xhatt}{\hat{x}_t}

\newcommand{\nxhat}{n_{\hat{x}}}

\newcommand{\up}{u^\mathrm{p}}
\newcommand{\uf}{u^\mathrm{f}}
\newcommand{\yp}{y^\mathrm{p}}
\newcommand{\wf}{w^\mathrm{f}}
\newcommand{\xp}{x^\mathrm{p}}

\newtheorem{problem}{Problem}
\newtheorem{definition}{Definition}
\newtheorem{assumption}{Assumption}
\newtheorem{remark}{Remark}
\newtheorem{thm}{Theorem}

\title{\LARGE \bf
System Identification for Virtual Sensor-Based Model Predictive Control: Application to a 2-DoF Direct-Drive Robotic Arm${}^{*}$
}

\author{Kosei Tsuji, Ichiro Maruta, Kenji Fujimoto, Tomoyuki Maeda, Yoshihisa Tamase, and Tsukasa Shinohara%
\thanks{${}^*$This work was partly supported by JSPS KAKENHI Grant Number JP24K00908 for the idea of the temporal weighting scheme.}%
\thanks{K. Tsuji, I. Maruta, and K. Fujimoto are with Department of Aeronautics and Astronautics, Graduate School of Engineering, Kyoto University, Kyotodaigaku-katsura, Nishikyo Ward, Kyoto City, Kyoto, 615-8540, Japan
        {\tt\small tsuji.kosei.68z@st.kyoto-u.ac.jp, maruta@kuaero.kyoto-u.ac.jp, k.fujimoto@ieee.org}}%
\thanks{T. Maeda, Y. Tamase, and T. Shinohara are with Kobe Steel, Ltd.,
        10-26, 2-chome Wakinohama-cho, Chuo-ward, Kobe City, Hyogo, 651-8585, Japan
        {\tt\small \{maeda.tomoyuki, tamase.yoshihisa, shinohara.tsukasa\}@kobelco.com}}%
}

\makeatletter
\def\ps@IEEEtitlepagestyle{%
   \def\@oddhead{%
    \hbox to\textwidth{%
      \hfil
      \raisebox{1\baselineskip}[0pt][0pt]{%
        \parbox{\textwidth}{%
          \centering\footnotesize
          This is the accepted version of a paper published in the proceedings of the
      2025 IEEE Conference on Decision and Control (CDC).\\
      The final published paper can be found at \href{https://doi.org/10.1109/CDC57313.2025.11313006}{doi:10.1109/CDC57313.2025.11313006}.
        }%
      }%
      \hfil
    }%
  }%

  \def\@oddfoot{\hfill
  \fbox{%
    \parbox{1.9\linewidth}{%
      \footnotesize
      \copyright\ 2025 IEEE. Personal use of this material is permitted. Permission from IEEE must be obtained for all other uses, in any current or future media, including reprinting/republishing this material for advertising or promotional purposes, creating new collective works, for resale or redistribution to servers or lists, or reuse of any copyrighted component of this work in other works.
    }%
  }%
  \hfill}%
  \def\@evenfoot{\@oddfoot}
}
\makeatother

\begin{document}
\maketitle

\thispagestyle{IEEEtitlepagestyle}
\pagestyle{empty}

\begin{abstract}
Nonlinear Model Predictive Control (NMPC) offers a powerful approach for controlling complex nonlinear systems, yet faces two key challenges.
First, accurately modeling nonlinear dynamics remains difficult. 
Second, variables directly related to control objectives often cannot be directly measured during operation.
Although high-cost sensors can acquire these variables during model development, their use in practical deployment is typically infeasible. 
To overcome these limitations, we propose a Predictive Virtual Sensor Identification (PVSID) framework that leverages temporary high-cost sensors during the modeling phase to create virtual sensors for NMPC implementation.
We validate PVSID on a Two-Degree-of-Freedom (2-DoF) direct-drive robotic arm with complex joint interactions, capturing tip position via motion capture during modeling and utilize an Inertial Measurement Unit (IMU) in NMPC. 
Experimental results show our NMPC with identified virtual sensors achieves precise tip trajectory tracking without requiring the motion capture system during operation. 
PVSID offers a practical solution for implementing optimal control in nonlinear systems where the measurement of key variables is constrained by cost or operational limitations.
\end{abstract}

\section{Introduction}
Nonlinear Model Predictive Control (NMPC) has become a widely used method for the practical optimal control of nonlinear systems due to its flexibility to optimize various performance specifications.
Despite its widespread adoption in industrial applications such as power electronics \cite{harbi2023model}, energy management \cite{byrne2017energy}, and mobile robotics \cite{jin2024physical}, two major challenges persist when applying NMPC to complex nonlinear systems.

The first challenge is the construction of accurate dynamical models for complex nonlinear systems whose state is unknown.
Although the method for identifying linear systems is mature, the field of nonlinear system identification is still in the process of development \cite{nelles2020nonlinear}.
The second challenge arises from the difficulty in obtaining measurements of the control variables that are directly related to the objective function \cite{darby2012mpc}.
For example, in a robotic arm system, the primary variable to be controlled is typically the arm tip position, but it is difficult to directly measure.
Instead, it is common to capture the behavior of the joint angles.
In order to convert joint angles to tip coordinates, a kinematic model must be constructed.

Regarding the first challenge, prior studies have investigated system identification based on nonlinear state-space models \cite{beintema2023deep, gedon2021deep, yamada2023subspace}, which are capable of capturing a wide range of nonlinear phenomena \cite{schoukens2019nonlinear}. 
Most existing state-space models (e.g., \cite{beintema2023deep, gedon2021deep}) are not designed specifically for multi-step prediction and therefore rely on recursive computation for multi-step forecasts, while \cite{yamada2023subspace} yields a multi-step prediction model.
Multi-step prediction models are robust to modeling errors \cite{kohler2022state} and are suitable for integration with NMPC.
Indeed, the MPC method in \cite{yamada2023subspace} can be interpreted as a natural nonlinear extension of Subspace Predictive Control (SPC), which is known for its computational efficiency \cite{fiedler2021relationship}.

As for the second challenge, reinforcement learning (RL) has emerged as a promising solution to address the gap between measurable outputs and control objectives, due to its ability to directly optimize policies based on task-oriented reward. 
However, safely conducting the numerous trials required for RL in real-world environments remains challenging. 
While Sim2Real approaches ---  which utilize simulations for extensive training prior to real-world deployment --- are widely adopted to address this limitation, they fundamentally depend on accurate simulation models, thus creating a substantial reliance on comprehensive prior modeling \cite{diprasetya2024sim, hwangbo2019learning}. 
Another important approach is the construction of virtual sensors, which replace high-cost measurements by learning from available data during the modeling phase.
For nonlinear systems, methods based on dynamic linearization \cite{zhang2022virtual}, affine parameter-varying approximations \cite{masti2021machine}, and Long Short-Term Memory (LSTM) networks \cite{yuan2019nonlinear} have been proposed.
These virtual sensors can reconstruct the outputs of high-cost sensors during deployment using only low-cost sensor data.
However, existing methods mainly focus on current state estimation and do not provide multi-step prediction models essential for MPC.

In this paper, 
we propose Predictive Virtual Sensor Identification (PVSID), a method that simultaneously identifies nonlinear system dynamics and virtual sensors as a multi-step prediction model. 
Based on the structure in \cite{yamada2023subspace}, our approach enables parallel multi-step output prediction, making it particularly suitable for NMPC.
In addition, we introduce a training scheme tailored to NMPC, which emphasizes accurate near-future prediction.
This enables control systems to operate without expensive sensors while maintaining performance comparable to
fully-instrumented setups.
The PVSID framework is particularly valuable as the identified models integrate seamlessly with existing NMPC implementations, enabling advanced control for complex systems under practical constraints.
we validate PVSID through tip-position trajectory tracking control of a Two-Degree-of-Freedom (2-DoF) direct-drive robotic arm using NMPC. 
Direct-drive robots eliminate backlash and friction issues but introduce complex joint interactions that are difficult to model accurately \cite{asada1984analysis}, making them ideal testbeds for our data-driven approach.

\section{Predictive virtual sensor identification (PVSID)}\label{sec:VSID}
In this section, we introduce Predictive Virtual Sensor Identification (PVSID), a method for simultaneously identifying both virtual sensors and dynamics models.
We first present the formal problem formulation, followed by our specific implementation approach  based on neural networks.

\subsection{Problem formulation}
We assume that the behavior of the target system can be expressed as the following finite-dimensional discrete-time nonlinear dynamics:
\begin{equation}\label{eq:state_transfer_eq}
    x_{t+1}=f(x_t,u_t),
\end{equation}
where $t\in \mathbb{Z}$ is the time index,
$x_t\in \mathbb{R}^{n_x}$ is the state, $u_t\in \mathbb{R}^{n_u}$ is the input, and $f : \mathbb{R}^{n_x} \times \mathbb{R}^{n_u} \to \mathbb{R}^{n_x}$ is the state-transition function.
Next, we consider the two distinct observation functions:
\begin{align}
    y_t &= h_y(x_t), \label{eq:measure_y}\\
    w_t &= h_w(x_t), \label{eq:measure_w}
\end{align}
where $h_y: \mathbb{R}^{n_x}\to \mathbb{R}^{n_y}$ and $h_w: \mathbb{R}^{n_x}\to \mathbb{R}^{n_w}$ are measurement functions that yields outputs with the following characteristics:
\begin{itemize}
\item $y$: Output variables that are consistently measurable during both model construction and system operation, but do not directly correspond to the control objectives (e.g., joint angles in a robotic arm)
\item $w$: Output variables that directly relate to the control objectives but can only be measured during model construction due to practical constraints (e.g., precise tip position of a robotic arm).
\end{itemize}

To formalize the PVSID problem, we first define time series vectors of past and future input/output data relative to time $t$:
\begin{equation}
\begin{aligned}
  \utp&\coloneqq \begin{bmatrix}
      u_{t-\hp}\\
      u_{t-\hp+1}\\
      \vdots\\
      u_{t-1}
  \end{bmatrix},&
  \utf&\coloneqq\begin{bmatrix}
      u_{t}\\
      u_{t+1}\\
      \vdots\\
      u_{t+\hf-1}
  \end{bmatrix}.
\end{aligned}\label{eq:input-output_vector}
\end{equation}
The corresponding past and future output time series vectors, $\ytp$ and $\wtf$, are defined analogously.
Using these definitions, we can now formulate the PVSID problem as follows:
\begin{problem}\label{prob:vsid}
Given the measured input-output data $\left\{\left(u_t, y_t, w_t\right)\right\}^{T+\hf-1}_{t=-\hp+1}$ and the design parameters $\hp, \hf, \nxhat \in \mathbb{N}$, construct a model consisting of:
\begin{itemize}
\item A state estimator $E_\phi : \left(\utp, \ytp \right) \mapsto \hat{x}_t$,
\item An output predictor $P_\theta : \left(\hat{x}_t, \utf\right) \mapsto \wtf$,
\end{itemize}
where $\hat{x}_t \in \mathbb{R}^{\nxhat}$ represents a state equivalent to $x_t$, i.e., there exists a map $x_t \mapsto \xhatt$, and $T \in \mathbb{N}$ indicates the dataset size.
\end{problem}
\begin{remark}
This problem formulation extends conventional system identification by simultaneously identifying both the system dynamics and virtual sensors for the otherwise unmeasurable output $w$.
\end{remark}

To analyze the conditions under which Problem \ref{prob:vsid} has a solution, we introduce some auxiliary functions.
For notational convenience, we define the following based on the state-transition function $f$ in \eqref{eq:state_transfer_eq} and a general measurement function $h$.
\begin{equation}
\begin{aligned}
    f^{k}\bigl(x_{t}, &[u_{t}^{\top}, \dots, u_{t + k - 1}^{\top}]^{\top}\bigr) \\
    &\coloneqq f\bigl(\cdots f\bigl(f(x_{t}, u_{t}),\,u_{t+1}\bigr)\cdots,\,u_{t + k - 1}\bigr),
\end{aligned}
\end{equation}
\begin{equation}
\begin{aligned}
    h^{k}\bigl(x_{t}, &[u_{t}^{\top}, \dots, u_{t + k - 1}^{\top}]^{\top}\bigr)\\
&\coloneqq
\begin{bmatrix}
h(x_{t})\\
h\circ f\bigl(x_{t}, u_{t}\bigr)\\
\vdots\\
h\circ f^{k - 1}\Bigl(x_{t}, \bigl[u_{t}^{\top}, \dots, u_{t + k - 2}^{\top}\bigr]^{\top}\Bigr)
\end{bmatrix}.
\end{aligned}
\end{equation}
We next introduce the concept of uniform k-observability for nonlinear systems:
\begin{definition}[uniform $k$-observability \cite{moraal1995observer}]
A system is uniformly $k$-observable if the mapping $\mathbb{R}^{n_x}\times\left(\mathbb{R}^{n_u}\right)^k \to \left(\mathbb{R}^{n_y}\right)^{k} \times \left(\mathbb{R}^{n_u}\right)^{k}$ by $ \left(x,\mathbf{u} \right) \mapsto \left(h^{k}(x,\mathbf{u}),\mathbf{u}\right)$
is injective.
\end{definition}
This leads to the following assumption regarding system observability:
\begin{assumption}\label{assumption:k-ob}
The system consisting of \eqref{eq:state_transfer_eq} and \eqref{eq:measure_y} is uniformly $\hp$-observable, and the system consisting of \eqref{eq:state_transfer_eq} and \eqref{eq:measure_w} is uniformly $\hf$-observable.
\end{assumption}
Additionally, we assume that:
\begin{assumption}\label{assumption:num_state}
The dimension of the state estimate is sufficient, such that $\nxhat \geq n_x$.
\end{assumption}
Under these assumptions, we can establish the following theorem regarding the existence of a solution to Problem \ref{prob:vsid}:
\begin{thm}\label{thorem:vsid}
Problem \ref{prob:vsid} has a solution under the Assumption \ref{assumption:k-ob} and \ref{assumption:num_state}.
\end{thm}
\begin{proof}
From Assumption \ref{assumption:k-ob}, the following mapping
\begin{equation}\label{eq:phiy-map}
  \Phi_y^{\hp} :
    \left\{\,\left(\up,\,h_y^{\hp}(\xp, \up)\right) 
      \bigm|
      \xp \in \mathbb{R}^{n_x},
      \up \in (\mathbb{R}^{n_u})^{\hp}
    \right\}
    \to \mathbb{R}^{n_x}
\end{equation}
that maps $\left(\utp, \ytp\right) \mapsto x_{t-\hp}$ exists.
Using the state transition function and $\Phi^{\hp}_y$, the following mapping
\begin{equation}\label{eq:psi-map}
  \Psi :
    \left\{\,\left(\up,\,h_y^{\hp}(\xp, \up)\right) 
      \bigm|
      \xp \in \mathbb{R}^{n_x},
      \up \in (\mathbb{R}^{n_u})^{\hp}
    \right\}
    \to \mathbb{R}^{n_x}
\end{equation}
that maps $\left(\utp, \ytp\right) \mapsto f^{\hp}\left(\Phi^{\hp}_y(\utp, \ytp), \utp\right)$ also exists.
Under Assumption \ref{assumption:num_state}, we can choose $\Psi$ and $h_w^{\hf}$ as the state estimator $E_\phi$ and the output predictor $P_\theta$, respectively. As a result, the existence of a solution to Problem \ref{prob:vsid} is guaranteed.
\end{proof}

\subsection{Neural network based modeling method}
\begin{figure}[bpt]
  \centering
  \includegraphics[width=0.7\linewidth]{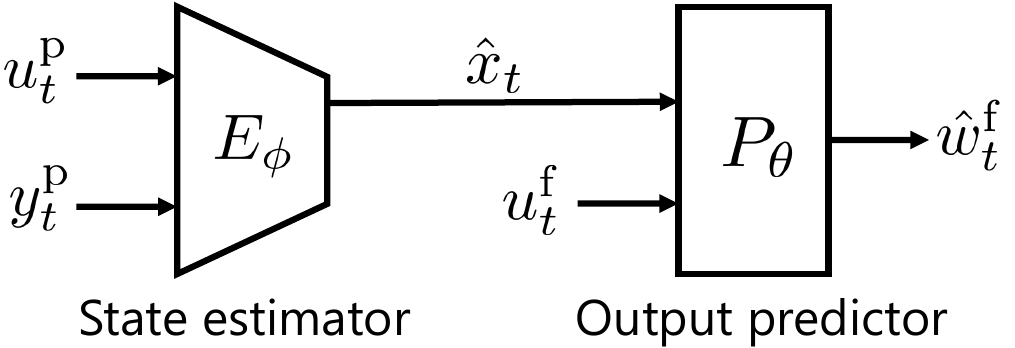}
  \caption{Structure of the neural network model for Predictive Virtual Sensor IDentification (PVSID). The sequence of past input-output $\left(\utp, \ytp\right)$ is reduced to current estimate $\xhatt$ and then reflected in the prediction of the sequence of future output $\wtf$ driven by $\utf$.}
  \label{fig:epnet}
\end{figure}

We describe a neural network architecture specifically designed for NMPC applications, where low computational cost and accurate near-future prediction are critical.

Figure~\ref{fig:epnet} illustrates the structure of our proposed PVSID model.
Both the state estimator $E_{\phi}$ and output predictor $P_{\theta}$ are implemented as neural networks.
Unlike recursive models, it predicts the entire horizon in a single forward pass, allowing for efficient gradient computation without recursive Jacobian propagation. This significantly reduces the computational cost of NMPC.

Furthermore, in this paper, we introduce a temporal weighting scheme to this loss function, which prioritizes accuracy on near-future predictions:
\begin{equation}\label{eq:epnet-optim_decay}
  \underset{\phi,\theta} {\text{min}} ~~~\frac{1}{T}\sum_{t=1}^T\sum_{k=1}^{\hf} \gamma^{k-1} \left\|\left(\wtf-P_{\theta}(E_{\phi}(\utp, \ytp), \utf)\right)_k\right\|^2.
\end{equation}
Here, $\phi$ and $\theta$ represent the neural network weight parameters for the state estimator and output predictor, respectively. 
The coefficient $\gamma\in(0,1]$ is a discount factor that reduces the influence of prediction errors as the time horizon extends, and $(\cdot)_k$ denotes the $k$-th time step component of the prediction error.
The introduction of a discount factor for non-recursive structures allows the model allocates its representational capacity according to prediction importance, prioritizing near-future accuracy without compromising learning efficiency.
This is particularly beneficial for nonlinear system identification, which often requires large amounts of data to accurately capture complex dynamics.
\begin{remark}
The ability to make this trade-off is not available in recursive prediction architectures, where improving short-term prediction accuracy inherently affects computational demands across the entire prediction horizon.
\end{remark}

The following theorem guarantees that the proposed model trained under \eqref{eq:epnet-optim_decay} is a solution to problem \ref{prob:vsid}.
\begin{thm}\label{thorem:vsid_NN}
Assuming `perfect training', i.e., for all $x^\mathrm{p} \in \mathbb{R}^{n_x}$, $\up \in (\mathbb{R}^{n_u})^{\hp}$, and $\uf \in (\mathbb{R}^{n_u})^{\hf}$, the equation:
\begin{equation}\label{eq:perfect_train}
    P_{\theta}(E_{\phi}(\up, \yp), \uf) = w^\mathrm{f},
\end{equation}
where $w^\mathrm{f}(x^\mathrm{p}, \up, \uf) \coloneqq h_w^{\hf}(x, \uf)$,
$x(x^\mathrm{p}, \up)\coloneqq f^{\hp}(x^\mathrm{p}, \up)$, and $\yp(x^\mathrm{p}, \up) \coloneqq h_y^{\hp}(x^\mathrm{p}, \up)$, holds.
Then, under Assumptions \ref{assumption:k-ob} and \ref{assumption:num_state}, the state estimate $\hat{x}$ is equivalent to $x$, i.e., there exists a map $E_\phi\left(\up, h_y^{\hp}(x^\mathrm{p}, \up)\right) \mapsto f^{\hp}(x^\mathrm{p}, \up)$.
\end{thm}
\begin{proof}
From Assumption \ref{assumption:k-ob}, the following mapping
\begin{equation}\label{eq:phiw-map}
  \Phi_w^{\hf} :
    \left\{\,\left(\uf,\,h_w^{\hf}(x, \uf)\right) 
      \bigm|
      x \in \mathbb{R}^{n_x},
      \uf \in (\mathbb{R}^{n_u})^{\hf}
    \right\}
    \to \mathbb{R}^{n_x}
\end{equation}
that maps $(\uf, \wf) \mapsto x$ exists.
Using the assumption of the perfect training \eqref{eq:perfect_train} and $\Phi_w^{\hf}$, the following mapping
\begin{equation}\label{eq:Gamma}
    \Gamma: \mathbb{R}^{n_{\hat{x}}} \times (\mathbb{R}^{n_u})^{\hf} \to \mathbb{R}^{n_x}
\end{equation}
that maps $(\hat{x}, \uf) \mapsto \Phi_w^{\hf} \left(\uf, P_\theta(\hat{x}, \uf) \right)$ exists.
By setting $\uf$ to a fixed vector value, such as the zero vector, $\Gamma$ reduces to a mapping from $\hat{x}$ to $x$.
\end{proof}
The universal-approximation theorem ensures that, given sufficient model capacity and training data, neural networks can achieve `perfect training' to arbitrary accuracy. Accordingly, the theorem should be regarded as applying only to systems whose behavior lies within the expressive power of the chosen architecture.
The practical plausibility of this `perfect training' assumption is supported--- at least in part ---by the real-world experiment reported in Section~\ref{sec:Experiments}.

\section{NMPC with PVSID-identified model}\label{sec:method_NMPC}
In this section, we demonstrate how the model identified through PVSID can be effectively applied to NMPC.
Our approach builds upon the methodology presented in \cite{yamasaki2024deep}.

The system state is estimated at each time step using the previously trained state estimator $E_\phi$ as $\hat{x}_t = E_\phi(\utp, \ytp)$.

Based on this state estimate, the control input is computed by solving an optimization problem, which minimizes a specified evaluation function over the prediction horizon:
\begin{align}  
\underset{\hat{{u}}_t^{\mathrm{f}}, \hat{{w}}_t^{\mathrm{f}}}{\operatorname{min}} \quad& L\left(\hat{{u}}_t^{\mathrm{f}}, \hat{w}_t^{\mathrm{f}}\right), &
\text { subject to } \quad& \hat{{w}}_t^{\mathrm{f}}=P_\theta\left(\hat{x}_t, \hat{{u}}_t^{\mathrm{f}}\right).
\end{align}

Here, the evaluation function  
$L: \left(\mathbb{R}^{n_u}\right)^{\hf} \times \left(\mathbb{R}^{n_y}\right)^{\hf}\to\mathbb{R}$ 
quantifies control performance, and is designed to address specific control objectives such as trajectory tracking while minimizing control effort.  
Note that constraints are not included in this formulation to focus on the PVSID methodology.

Since NMPC requires solving optimization problems at each step, computational efficiency is crucial for real-time implementation. 
To simplify the computational procedure while maintaining sufficient expressivity for our case study, we formulate the evaluation function in a quadratic form:
\begin{equation}\label{eq:objective_func_NMPC}  
L\left(\hat{{u}}_t^{\mathrm{f}}, P_\theta\left(\hat{\boldsymbol{x}}_t, \hat{{u}}_t^{\mathrm{f}}\right)\right)=\frac{1}{2} \ell\left(\hat{{u}}_t^{\mathrm{f}}\right)^{\top} \ell\left(\hat{{u}}_t^{\mathrm{f}}\right)  
\end{equation}  
where $\ell$ is a vector-valued function designed according to specific control objectives. 
This nonlinear least-squares structure can be efficiently solved using the Levenberg-Marquardt (LM) method \cite{levenberg1944method, marquardt1963algorithm}, with the update rule:
\begin{equation}  
\hat{u}_t^{\mathrm{f}}[l+1] = \hat{u}_t^{\mathrm{f}}[l] - \bigl(J_t[l]^{\intercal}J_t[l] + \lambda \mathbb{I}\bigr)^{-1} J_t[l]^\intercal\ell(\hat{u}_t^{\mathrm{f}}[l]),  
\end{equation}
where, $[l]$ denotes the iteration index of the optimization, and $J_t[l]$ is the Jacobian matrix of the error $\ell(\hat{u}_t^\mathrm{f})$ at the $l$-th iteration, and $\lambda$ controls the step size.  

This update is performed up to the maximum number of iterations per control cycle.
Once the optimization has finished, the first entry of $\hat{u}_t^\mathrm{f}$ is applied and the subsequent control cycle starts.

\begin{remark}
In our implementation, the neural network architecture enables efficient computation of the Jacobian through automatic differentiation, and the optimization is initialized with the solution from the previous control cycle to improve convergence speed.
\end{remark}

\section{Experiments}\label{sec:Experiments}

This section validates the proposed PVSID method through experiments on a 2-DoF direct-drive robotic arm.
Direct-drive robot arms exhibit complex nonlinear dynamics due to significant interference torques between links, making them challenging control targets.
The primary control objective is precise tip position tracking, while the tip position itself can only be measured using high-cost systems like motion capture.
This scenario perfectly aligns with PVSID's problem formulation.

\subsection{Experimental setup}

\begin{figure*}
    \centering
    \includegraphics[width=0.90\linewidth]{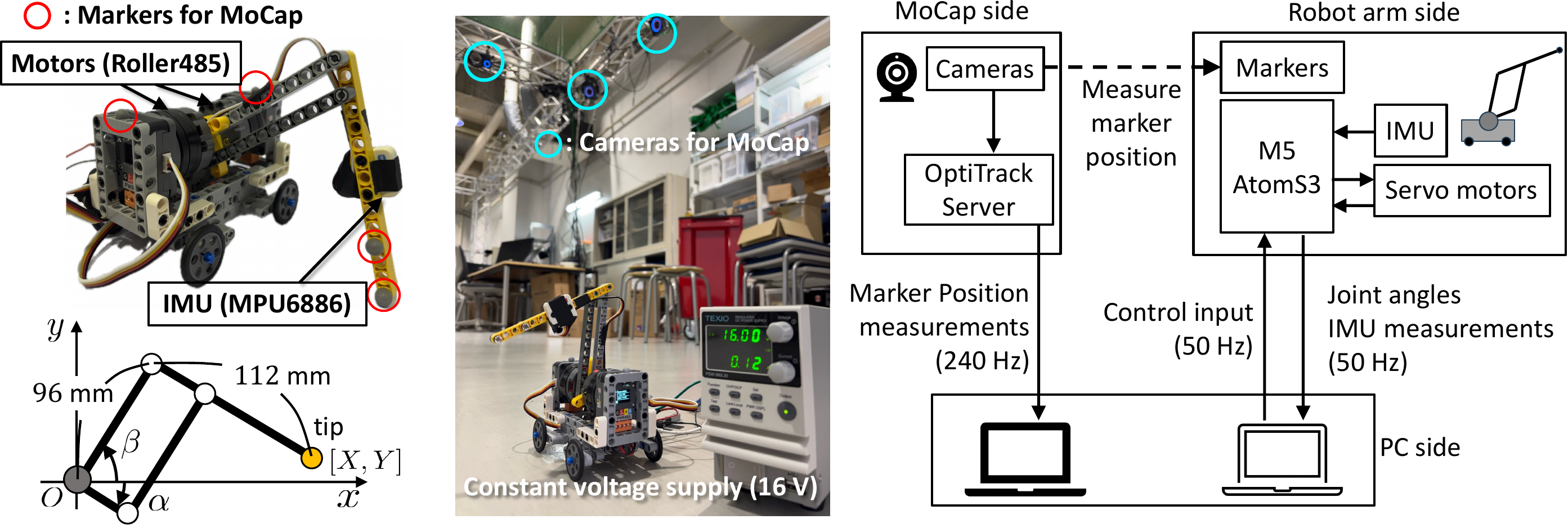}
    \caption{Overview of the experimental setup. Left: The developed 2 DoF direct-drive robotic arm and the schematic diagram. Middle: Experimental device configuration, including the motion capture system (MoCap) and controlled hardware. Right: Communication architecture between components.}
    \label{fig:experimental_setup}
\end{figure*}

As shown in Fig.~\ref{fig:experimental_setup}, we used a 2-DoF direct-drive robotic arm built from LEGO parts, driven by two parallel Roller 485 actuators in position control mode.
An Inertial Measurement Unit (IMU, MPU6886) mounted on the arm was used to detect subtle vibrations due to structural looseness inherent in LEGO construction and undetectable from joint angle alone. 
The experiment was conducted at the motor's rated voltage of \SI{16}{V}.

During model construction, precise tip position data were collected using an OptiTrack motion capture system.
The control system operates at \SI{50}{Hz}, with an M5Stack-ATOMS3 microcontroller serving as interface between the external PC and both the actuators and the IMU.
This setup provides joint angles, angular velocity, and acceleration data ---the $y$ variables in the PVSID framework representing the sensory information available during both training and operation.

\subsection{Experimental configuration}
The robot coordinate system and structure are illustrated in the left of Fig.~\ref{fig:experimental_setup}.
The direct and inverse kinematics relating joint angles $[\alpha, \beta]$ to the tip position $[X, Y]$ follow standard formulations for a 2-DoF planar manipulator, with physical constraints restricting the configuration to a specific solution branch.

\subsection{Data collection for PVSID}
For the PVSID implementation, 
we defined the tip position as $w$, obtained from motion capture and expressed in the robot coordinate system, 
and the combination of joint angles and IMU data (angular velocity and acceleration) as $y$.
The control input $u$ consists of reference joint angles for the position-controlled motors. 

To generate identification data, we 
constructed a single, continuous trajectory by sequentially connecting straight-line segments between randomly selected target tip positions and tip velocities (\SIrange{30}{300}{mm/s}) within the region reachable by the tip, as determined by the geometric structure of the arm.
The arm followed this trajectory, which effectively explored the operational space and ensured non-periodic, diverse motions.
During data collection, the reference tip coordinates along the trajectory were converted into reference joint angles (control input $u$) using inverse kinematics.

Data was collected for one hour at \SI{50}{Hz} sampling rate, while the tip position data, initially acquired at \SI{240}{Hz}, was resampled to \SI{50}{Hz} to align with the control period for simplicity.
The collected dataset was split into training, validation, and test sets, with dataset sizes of approximately \num{110000}, \num{40000}, and \num{40000} samples, respectively.

\subsection{Neural network training and validation}

\begin{figure}[t]
\centering
\includegraphics[width=0.9\linewidth]{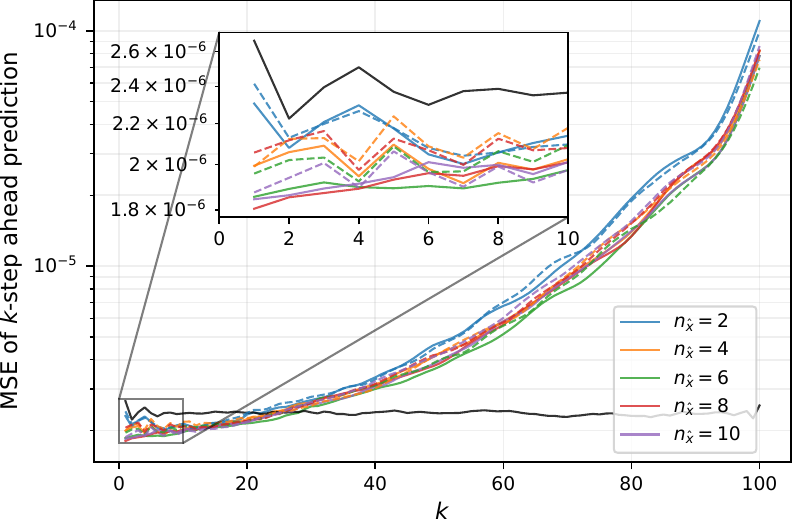}
\caption{
Mean squared error  of $k$-step ahead predictions for test data. 
The black line shows the model with $\gamma = 1.0$, $n_{\hat{x}} = 8$, and IMU. 
Other $10$ lines represent models trained with $\gamma = 0.9$ for $n_{\hat{x}} \in \{2, 4, 6, 8, 10\}$, with and without IMU.
The solid and dotted lines correspond to the results with and without IMU, respectively.
}\label{fig:compare_prediction_horizon}
\end{figure}

\begin{figure}[t]
\centering
\includegraphics[width=0.9\linewidth]{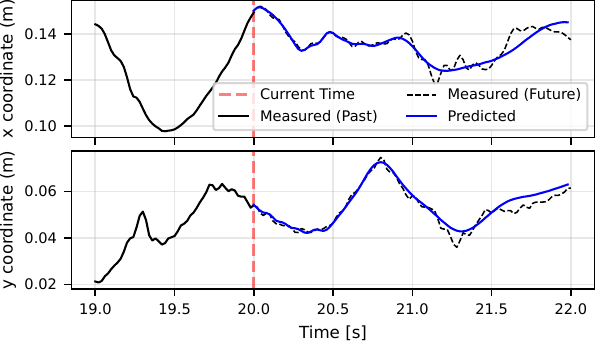}
\caption{Model prediction results at $t=\SI{20}{s}$ in the test dataset.
The red dotted line indicates the current time.
The black solid line shows the past actual output, the black dotted line shows the future actual output, and the blue line shows the predicted future output.
}
\label{fig:sanpshot_prediction}
\end{figure}

The state estimator $E_{\phi}$ and the output predictor $P_{\theta}$ are implemented as feedforward neural networks with three hidden layers of $32$ and $256$ neurons, respectively; all layers use the ReLU activation function.  
The past and future horizons are set to $\hp = 50$ and $\hf = 100$, respectively.
The model was implemented using Python/PyTorch and trained on a computer with an AMD Ryzen 9 9950X and an NVIDIA GeForce RTX 4090.
To prevent overfitting, the AdamW \cite{loshchilov2019} optimizer was used with the default parameters provided by PyTorch.
Training ran for \num{100000} epochs with the final model selected based on the lowest validation loss computed by \eqref{eq:epnet-optim_decay}.
The total training time was approximately $40$ minutes.

To evaluate how the inclusion of IMU data improves arm dynamics modeling, we trained ten models with different configurations, comparing models using only joint angles versus models using both joint angles and IMU data, while varying the dimension of state estimate $n_{\hat{x}}$ in $\{2,4,6,8,10\}$.
Fig.~\ref{fig:compare_prediction_horizon} shows how the Mean Squared Error (MSE) of $k$-step ahead prediction:
\begin{equation}
    \frac{1}{M}\sum_{t=1}^{M}\left\|w_{t+k-1}-\left(\hat{w}^{\mathrm{f}}_t\right)_{k}\right\|^2 \quad (k = 1, 2, \cdots, \hf),
\end{equation}
where $M$ is the test dataset size, increases with the prediction step $k$, which aligns with our MPC-focused design prioritizing near-future accuracy.
The inset of Fig.~\ref{fig:compare_prediction_horizon} highlights the short-term prediction performance, which is crucial for effective MPC implementation,
With IMU data, significant MSE reduction occurs when $n_{\hat{x}} \geq 6$,  suggesting that the IMU effectively captures subtle arm dynamics.
In contrast, models without IMU data show modest improvement only at $n_{\hat{x}} = 10$, indicating that even with increased state dimension, these models struggle to capture the complex dynamics that IMU measurements readily provide.

We selected the model with IMU and $n_{\hat{x}} = 8$ for optimal short-term prediction accuracy.
Fig.~\ref{fig:sanpshot_prediction} illustrates test data predictions for this model, demonstrating its ability to accurately predict tip position, particularly in the near-future.
Furthermore, the black line in Fig. ~\ref{fig:compare_prediction_horizon} shows the performance of the model trained with $\gamma=1.0$ (i.e., without temporal discounting) using IMU data and $n_{\hat{x}}=8$. Compared to this, models trained with $\gamma=0.9$ achieve lower MSE values, indicating that the the discount factor improves short-term prediction performance.

\subsection{NMPC configuration}
We implemented NMPC with virtual sensors described in Section \ref{sec:method_NMPC}.
The error function $\ell$ in the objective function \eqref{eq:objective_func_NMPC} is designed as:
\begin{align}
\ell(u_t^{\mathrm{f}}) &= \begin{bmatrix}
\left(P_\theta\left(\xhatt, u_t^{\mathrm{f}}\right) - r_t^\mathrm{f}\right)/w_{\mathrm{std}}\\
\zeta\Delta(\uf_t)/u_{\mathrm{std}}
\end{bmatrix}\in \mathbb{R}^{\hf(n_\mathrm{u}+n_\mathrm{w})},\\
\Delta(\uf_t) &\coloneqq \left(\begin{bmatrix}
    \left(\up_t\right)_{\hp}\\
    \left(\uf_t\right)_{1:(\hf-1)}
\end{bmatrix} - u_t^{\mathrm{f}}\right)\in \left(\mathbb{R}^{n_\mathrm{u}}\right)^{\hf},
\end{align}
where $r_t^\mathrm{f} \in \left(\mathbb{R}^{n_w}\right)^{\hf}$ is the reference output, $\left(\uf_t\right)_{1:(\hf-1)}$ denotes the $\uf_t$ without the last time step component, and $\Delta(\uf_t)$ measures the difference between consecutive control input across time steps, effectively penalizing rapid change in reference position.
The normalization factors $w_{\mathrm{std}}$ and $u_{\mathrm{std}}$ are the standard deviations of output $w$ and input $u$ from the training data.
The weighting parameter $\zeta=0.5$ balances tacking accuracy against input smoothness, with higher values producing more gradual control actions.

The optimization was solved using the LM method, as described in Section \ref{sec:method_NMPC}, with $\lambda=\SI{1e-7}{}$, which is empirically determined,
on a laptop PC with Intel Core i9-14900HX and NVIDIA GeForce RTX 4090 Laptop GPU.
Each iteration took approximately \SI{6}{ms}.

\subsection{Trajectory tracking results}

\begin{figure}[t]
\centering
\includegraphics[width=0.9\linewidth]{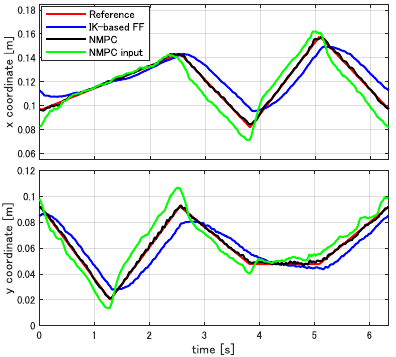}
\caption{Time series of tip position coordinates during star-shaped trajectory tracking. 
The red line represents the reference tip position, which also corresponds to the input of inverse kinematics-based feedforward control (IK-based FF) after conversion to tip coordinates via forward kinematics.
The black and blue lines show the measured tip positions under NMPC control and IK-based FF, respectively. 
The green line depicts the control input generated by NMPC, expressed in tip coordinates using forward kinematics.
The tip position MSE under NMPC control was significantly lower at $(\SI{1.3}{mm})^2$, compared to $(\SI{11.5}{mm})^2$ with IK-based FF control.
}\label{fig:star_traj_time_trans}
\end{figure}

To validate the effectiveness of the proposed method, we conducted a star-shaped trajectory tracking experiment comparing NMPC with a model identified through PVSID against a feedforward control approach based on inverse kinematics (IK-based FF), where the target joint angles were obtained via inverse kinematics from the desired tip trajectory. 
The experiment was performed at a tip velocity of \SI{60}{mm/s}.

Figure~\ref{fig:star_traj_time_trans} presents the results for one complete cycle of the star-shaped trajectory. 
The blue line depicts the tip position under a simple inverse kinematics-based approach where target joint angles calculated from desired tip positions were directly fed to the actuators, relying solely on the actuators' internal controllers.
The red line represents the target tip position, the black line shows the measured tip position under NMPC control,
and the green line depicts the control input by NMPC, that is the reference joint angles,  converted to tip coordinates through forward kinematics for visualization purposes.

The significant time lag in the blue trajectory relative to the target demonstrates the substantial actuator dynamics and delays inherent in the system.
This highlights the necessity of incorporating a predictive feedforward term based on an accurate system model.
In contrast, the close alignment between the NMPC trajectory and the target trajectory clearly demonstrates that the proposed NMPC successfully compensates for system dynamics and kinematic errors, enabling precise trajectory tracking. 
The subtle high-frequency oscillations seen in both trajectories are due to the vibrations form structural looseness.
Notably, the green line reveals that the NMPC strategy deliberately introduces low-frequency oscillatory patterns in the reference inputs to counteract complex nonlinear interlink dynamics.

This demonstrates that the PVSID-identified model captured complicated interlink dynamics and structural characteristics specific to direct-drive robots—complex behaviors that would be challenging to model through traditional approaches.

\section{Conclusion}\label{sec:conclusion}

This paper introduced Predictive Virtual Sensor Identification (PVSID), a method addressing key challenges in applying Nonlinear Model Predictive Control (NMPC) to nonlinear systems, particularly 
the difficulty in measuring variables directly related to control objectives. 
PVSID leverages high-cost sensors only during the modeling phase and replaces them with virtual sensors capable of accurate predictions in actual operation.

Experimental validation on a 2-DoF direct-drive robotic arm demonstrated that PVSID effectively captured complex nonlinear dynamics, including joint interactions and structural characteristics. 
Integrating the identified model into NMPC allowed precise tip trajectory tracking without the operational use of expensive motion capture systems.
The results indicate that PVSID effectively leverages diverse sensor data to anticipate and compensate for complex system behaviors.

Future work will explore the integration of input and output constraints into NMPC formulations, as well as extend the approach to higher-dimensional systems such as multi-DoF robotic arms, to further evaluate scalability and robustness in more complex control scenarios.

\bibliographystyle{IEEEtran}
\bibliography{ref.bib}

\end{document}